\documentclass[preprint,3p]{elsarticle}
\usepackage{tabularx}
\usepackage{colortbl}
\usepackage{amsthm}
\usepackage{amsmath,amssymb,enumerate,subfigure,epsfig,psfrag,hhline,srcltx}

\usepackage{color}
\usepackage{xcolor}
\usepackage{hyperref}
\hypersetup{breaklinks=true,colorlinks=false,citecolor=false}
\usepackage{graphicx}          % Include this line if your
\usepackage[noabbrev,capitalise,nameinlink]{cleveref}

\biboptions{sort&compress}

\newdefinition{example}{Example}

\newdefinition{remark}{Remark}

\newtheorem{lemma}{Lemma}
\newtheorem{definition}{Definition}
\newtheorem{proposition}{Proposition}

\allowdisplaybreaks

\begin{document}
\begin{frontmatter}
\title{Harnessing Membership Function Dynamics for Stability Analysis of T-S Fuzzy Systems}

\author[KAIST]{Donghwan Lee}\ead{donghwan@kaist.ac.kr}
\author[Hanbat]{Do Wan Kim}\ead{dowankim@hanbat.ac.kr}

\address[KAIST]{Department of Electrical Engineering, KAIST, Deajeon, South Korea}
\address[Hanbat]{Department of Electrical Engineering, Hanbat National University, Daejeon, South Korea}

\begin{abstract}
The main goal of this paper is to develop a new linear matrix inequality (LMI) condition for the asymptotic stability of continuous-time Takagi-Sugeno (T-S) fuzzy systems. A key advantage of this new condition is its independence from the bounds on the time-derivatives of the membership functions, a requirement present in the existing approaches. This is achieved by introducing a novel fuzzy Lyapunov function that incorporates an augmented state vector. Notably, this augmented state vector encompasses the membership functions, allowing the dynamics of these functions to be integrated into the proposed condition. This inclusion of additional information about the membership function serves to reduce the conservativeness of the suggested stability condition. To demonstrate the effectiveness of the proposed method, examples are provided.
\end{abstract}


\end{frontmatter}

\section{Introduction}

Over the last few decades, Takagi-Sugeno (T-S) fuzzy systems~\cite{song2023improved,sun2023composite,du2024dynamic,nguyen2019fuzzy,xie2022relaxed,xie2022enhanced,tian2021switched,zheng2021membership,pan2021improved,ku2021observer,zhu2020adaptive} have emerged as popular tools for modeling nonlinear systems~\cite{khalil2002nonlinear}. These systems are characterized by a convex combination of linear subsystems, each modulated by nonlinear membership functions (MFs)~\cite{tanaka2004fuzzy}. Notably, within the context of Lyapunov theory, T-S fuzzy systems provide a systematic method to apply linear system theories to the stability analysis and control design of nonlinear systems. Additionally, they allow for the development of problem formulations based on convex linear matrix inequality (LMI) optimization methods, which can be efficiently solved using standard convex optimization techniques~\cite{boyd1994linear}. Due to these benefits, there has been significant interests in the stability analysis and control design of T-S fuzzy systems during the last decades~\cite{nguyen2019fuzzy}.

While T-S fuzzy model-based methods offer advantages, they are often hampered by a significant level of conservatism arising from various sources: 1) The transformation of infinite-dimensional parameterized linear matrix inequalities (PLMIs), inherent in Lyapunov inequalities, into finite-dimensional LMI problems introduces notable conservatism. 2) The conservative nature also stems from the limited architectures of Lyapunov functions and control laws. 3) The effort to ensure stability across a broad region of the state-space can lead to excessive conservatism, especially when the actual nonlinear system is stable only in a limited region around the origin. To mitigate this conservatism, extensive research has been dedicated, which generally falls into three main categories:
\begin{figure}
\centering
\includegraphics[width=6cm,height=5.5cm]{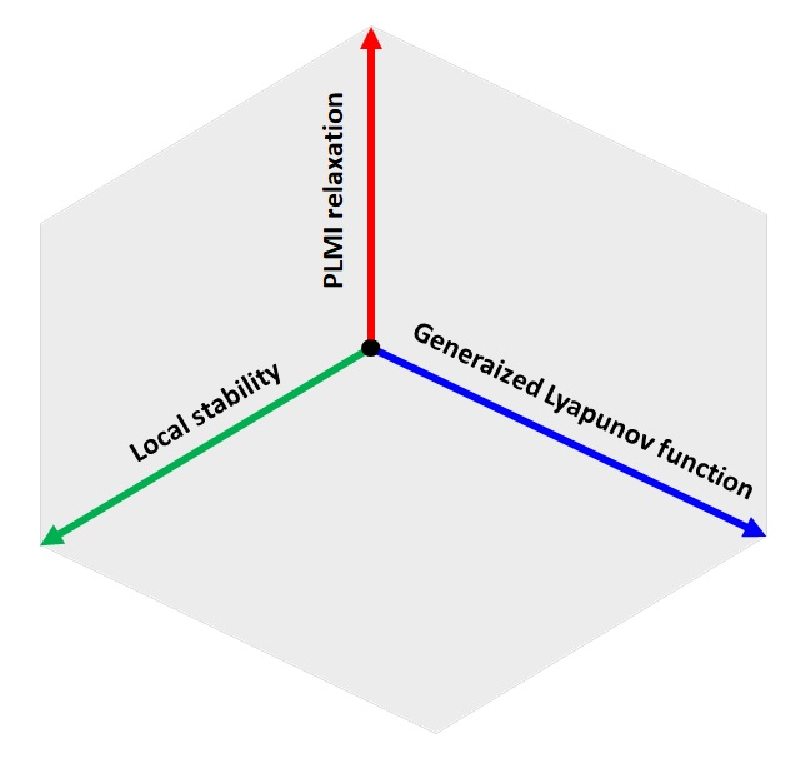}
\caption{Three main directions to reduce the conservatism}\label{fig:1}
\end{figure}
\begin{enumerate}
\item {\bf Relaxation of PLMIs}: Techniques for relaxing PLMIs have been developed, which include various methods such as: 1) slack variable methods, as explored in several studies~\cite{fang2006new,kim2000new,kim2022improvement}, which introduce additional variables to reduce conservatism in the LMIs, 2) convex summation relaxation techniques~\cite{tuan2001parameterized,kim2023relaxed}, which offer a way to approximate the solution of PLMIs more accurately, 3) multidimensional convex summation techniques~\cite{sala2007asymptotically,oliveira2007parameter} based on Polya’s theorem, which extend the convex summation approach into higher dimensions for improved results, 4) methods~\cite{sala2007relaxed,narimani2009relaxed,kruszewski2009triangulation,lam2009quadratic} that utilize the structural information of MFs, providing a more tailored approach to relaxing PLMIs by considering the specific properties of the MFs.

\item {\bf Generalization of Lyapunov Functions and Control Laws}: Numerous advancements have been made in the realm of traditional Lyapunov functions, leading to the proposal of various extensions. These include piecewise Lyapunov functions, as discussed in several studies~\cite{johansson1999piecewise,feng2003controller,campos2012new}. Additionally, the concept of fuzzy Lyapunov functions (FLFs) has been extensively explored~\cite{tanaka2003multiple,mozelli2009reducing,guerra2004lmi,ding2006further,lee2010improvement,lee2011new,lee2014generalized,lee2013relaxed,ding2010homogeneous,lazarini2021relaxed}. Other notable extensions are line integral Lyapunov functions~\cite{rhee2006new,gonzalez2019generalised}, polynomial Lyapunov functions~\cite{Tanaka2009,sala2009polynomial,lam2011polynomial,bernal2011stability}, and augmented FLFs~\cite{kruszewski2008nonquadratic,lee2011approaches,lee2010further,lee2011further}, among others.

\item {\bf Local Stability Approaches}: Several studies have concentrated on determining the domain of attraction (DA) via local stability analysis~\cite{lee2013relaxed,bernal2011stability,bernal2010generalized,pan2011nonquadratic,lee2012fuzzy,lee2013linear,lee2016local,lee2014local}. These methods assess stability within a confined region, aiming to strike a balance between the size of the DA and the level of conservatism. This approach seeks to minimize conservatism while maintaining a reasonable estimation of the DA.

\end{enumerate}

The main goal of this paper is to introduce a new, derivative-free, and less conservative LMI criterion for the asymptotic stability of continuous-time T-S fuzzy systems based on an augmented Lyapuov function. While the FLF method is widely recognized for its effectiveness in reducing conservativeness, it generally requires bounding the time-derivatives of the MFs. However, it introduces a tuning issue for these bounds. Moreover, this method necessitates considering local stability, as only smaller state-space subsets meet the time-derivative constraints, which requires additional efforts for complex analysis.

Motivated by the above discussions, we propose a new LMI-based condition, which eliminates the need for such bounds on the derivatives. We expect that this new approach makes a meaningful departure from the existing approaches and provides additional insights into stability of T-S fuzzy systems. This advancement is achieved by introducing a new Lyapunov function, which includes an augmented state vector. This augmented state vector encompasses the MFs, which enables the integration of the MF dynamics into our proposed criterion, and this additional information of the system dynamics potentially helps reduce conservativeness of the condition.

In other words, we treat the T-S fuzzy system as a nonlinear system which is an interconnection of two nonlinear subsystems, one captures the dynamics of the membership functions (MFs) and the other represents the original nonlinear system.
This approach not only offers new insights into stability analysis but also enhances the resulting stability condition in terms of reduced conservativeness and increased practicality. To validate the effectiveness of our method, some illustrative examples are given.

In summary, the main contribution of this paper is the introduction of a new LMI criterion for the asymptotic stability of continuous-time T-S fuzzy systems, featuring a newly proposed augmented Lyapunov function that eliminates the need for bounding the derivatives of MFs. In particular, the approach to reducing conservativeness in this study is summarized as follows:
\begin{enumerate}
\item The proposed method integrates MFs as new state variables, incorporating their dynamics into the nonlinear system under consideration. By treating the T-S fuzzy system as an interconnected nonlinear system composed of two subsystems, one representing MF dynamics and the other the original nonlinear system, we obtain additional insights into the system dynamics that reduce the conservativeness of the stability condition.

\item From this analysis, the proposed condition eliminates the need for bounding the derivatives of MFs by integrating them into the nonlinear system dynamics, and hence, it reduces the conservatism induced by such bounds.
\end{enumerate}

\section{Preliminaries}\label{sec:preliminaries}

\subsection{Notation}

The adopted notation is as follows. ${\mathbb R}$: sets of real
numbers; $A^T$: transpose of matrix $A$; $A \succ 0$ ($A \prec 0$,
$A \succeq 0$, and $A \preceq 0$, respectively): symmetric
positive definite (negative definite, positive semi-definite, and
negative semi-definite, respectively) matrix $A$; $I_n$ and $I$: $n\times n$ identity matrix and identity matrix of appropriate dimensions, respectively; $0$: zero vector of appropriate dimensions; ${\rm{co}}\{  \cdot \}$: convex hull \cite{Boyd2004}; $*$ inside a matrix: transpose of its symmetric term; ${\bf{1}}_r : = [\begin{array}{*{20}c}
   1 & 1 &  \cdots  & 1  \\
\end{array}]^T \in {\mathbb R}^r $; $A \otimes B$: Kronecker product of matrices $A$ and $B$.

\subsection{Takagi--Sugeno (T--S) fuzzy system}

In this paper, we consider the continuous-time nonlinear system
\begin{align}
&\dot x_t = f(x_t),\label{eq:nonlinear-system}
\end{align}
where $t \geq 0$ is the time, $x_t: = \begin{bmatrix} x_{t,1} & \cdots  & x_{t,n}\\
\end{bmatrix}^T \in {\mathbb R}^n$ is the state, $f:{\mathbb R}^n \to {\mathbb R}^n$ is a nonlinear function
such that $f(0) = 0$, i.e., the origin is an equilibrium
point of~\eqref{eq:nonlinear-system}. By the sector nonlinearity
approach~\cite{tanaka2004fuzzy}, one can obtain the following T--S
fuzzy system representation of the nonlinear
system~\eqref{eq:nonlinear-system} in the subset ${\cal X}$ of the
state-space ${\mathbb R}^n$:
\begin{align}
\dot x_t = \sum_{i = 1}^r {{\alpha _i}(x_t){A_i}x_t} =: A(\alpha(x_t)) x_t ,\quad \forall x_t \in {\cal X}\label{eq:fuzzy-system}
\end{align}
where $A_i \in {\mathbb R}^{n \times
n}, i \in {\cal I}_r :=
\{ 1,\,2, \ldots ,\,r\}$ are subsystem matrices, and $\alpha_i: {\cal X} \to {\mathbb R}, i \in {\cal I}_r$ is called a membership function (MF).
The following vector enumerates the MFs:
\begin{align}
\alpha (x): = \left[ {\begin{array}{*{20}{c}}
{{\alpha _1}(x)}\\
{{\alpha _2}(x)}\\
 \vdots \\
{{\alpha _r}(x)}
\end{array}} \right] \in {\mathbb R}^r\label{eq:alpha-vec}
\end{align}
that lies in the unit simplex
\[
{\Lambda _r}: = \left\{ {\begin{array}{*{20}{c}}
{\alpha  \in {\mathbb R}^r: \sum\limits_{i = 1}^r {{\alpha _i}}  = 1,0 \le {\alpha _i} \le 1,i \in {\cal I}_r}
\end{array}} \right\}.
\]
In addition, ${\cal X} \in {\mathbb R}^n $ is a subset of the state-space ${\mathbb R}^n$ including the origin $0 \in {\cal X}$ and satisfies
\begin{align}
{\cal X} \subseteq \left\{ {x \in {\mathbb R}^n:f(x) = \sum_{i = 1}^r {{\alpha _i}(x)A_ix} } \right\}\label{region-X}
\end{align}
In other words, ${\cal X}$ is a modelling region where the T--S fuzzy model is defined.

As a running example, we consider the following system.
\begin{example}[{\cite{mozelli2009reducing}}]\label{ex:1}
Let us consider the nonlinear system~\eqref{eq:nonlinear-system} with $f(x) = \left[ {\begin{array}{*{20}{c}}
{{x_2}}\\
{\frac{{ - 4 - \lambda  + \lambda \sin {x_1}}}{2}{x_1} - {x_2}}
\end{array}} \right]$, which can be rewritten by the fuzzy system~\eqref{eq:fuzzy-system} with
\begin{align}
&A_1  = \left[ {\begin{array}{*{20}c}
   { 0} & 1  \\
   { - 2} & { - 1}  \\
\end{array}} \right],\quad A_2  = \left[ {\begin{array}{*{20}c}
   { 0} & 1  \\
   { - \left( {2 + \lambda } \right)} & { - 1}  \\
\end{array}} \right],\label{system-matrix}\\
&\alpha_1(x)  = \frac{{1 + \sin x_1 }}{2},\quad \alpha_2(x)  = 1 -
\alpha_1(x),\nonumber\\
&{\mathcal X} = \left\{ {x \in {\mathbb R}^n :\,x_1  \in
[ - \pi /2,\,\pi /2]} \right\}\nonumber
\end{align}
\end{example}

In what follows, some important notions and lemmas are briefly introduced.
\begin{definition}
The system~\eqref{eq:nonlinear-system} or~\eqref{eq:fuzzy-system} is called locally asymptotically stable at the origin if there exists a neighborhood of $0 \in {\mathbb R}^n$ such that any $x_t$ with $x_0$ inside the neighborhood, $x_t \to 0$ as $t\to \infty$.  A region ${\cal D} \subseteq {\cal X}$ is called a domain of attraction (DA)~\cite{khalil2002nonlinear} if any $x_t$ with $x_0 \in {\cal D}$, $x_t \to 0$ as $t\to \infty$.
\end{definition}

\begin{definition}
For any symmetric matrix $\Upsilon(\alpha) \in {\mathbb R}^{n\times n}$ depending on the parameter $\alpha \in \Lambda_r$, the following matrix inequality is called a parameterized matrix inequality (PMI):
\begin{align*}
\Upsilon (\alpha ) \prec 0,\quad \alpha  \in {\Lambda _r}.
\end{align*}
When the PMI depends linearly on the decision variables, then it is called the parameterized linear matrix inequality (PLMI).
Let us consider any symmetric matrices ${\Upsilon _{{i_1}{i_2} \cdots {i_N}}},({i_1},{i_2}, \ldots ,{i_N}) \in {\cal I}_r \times {\cal I}_r \times  \cdots  \times {\cal I}_r$.
A PMI represented by
\begin{align*}
\Upsilon (\alpha ): = \sum\limits_{{i_1} = 1}^r {\sum\limits_{{i_2} = 1}^r { \cdots \sum\limits_{{i_N} = 1}^r {{\alpha _{{i_1}}}{\alpha _{{i_2}}} \cdots {\alpha _{{i_N}}}{\Upsilon _{{i_1}{i_2} \cdots {i_N}}}} } }  \prec 0
\end{align*}
for all $\alpha  \in {\Lambda _r}$ is called the multi-dimensional or $N$-dimensional fuzzy summation.
When $N=1$, it is called the single fuzzy summation, while when $N=2$, it is called the double fuzzy summation.
\end{definition}

PMIs or PLMIs are infinite-dimensional conditions, which pose challenges for numerical verification. Over recent decades, numerous techniques have been developed to transform them into finite-dimensional sufficient LMI conditions. Specifically, for PLMIs involving single fuzzy summations, the subsequent LMI relaxation technique can be employed.
\begin{lemma}[\cite{wang1996approach}]\label{thm:relax0}
Let us consider any symmetric matrices $\Upsilon _{i} \in {\mathbb R}^{n\times n},i \in {\cal I}_r$.
The following PMI holds:
\begin{align*}
\Upsilon (\alpha ): = \sum_{i = 1}^r {{\alpha _i}\Upsilon _{i} }  \prec 0,\quad \forall \alpha  \in {\Lambda _r}
\end{align*}
if $\Upsilon _{i} \prec 0$ holds for all $i \in {\cal I}_r$.
\end{lemma}

For PLMIs that incorporate double fuzzy summations, the following LMI relaxation technique can be applied.
\begin{lemma}[\cite{tanaka1998fuzzy}]\label{thm:relax1}
Let us consider any symmetric matrices $\Upsilon _{ij} \in {\mathbb R}^{n\times n},(i,j) \in {\cal I}_r \times {\cal I}_r$.
The following PMI holds:
\begin{align*}
\Upsilon (\alpha ): = \sum_{i = 1}^r {\sum_{j = 1}^r {{\alpha _i}{\alpha _j}{\Upsilon _{ij}}} }  \prec 0,\quad \forall \alpha  \in {\Lambda _r}
\end{align*}
if $\Upsilon _{ij} + \Upsilon _{ji} \prec 0$ holds for all $\forall (i,j) \in {\cal I}_r \times {\cal I}_r, i \ge j$.
\end{lemma}

\subsection{Quadratic stability}

Let us consider the quadratic Lyapunov function (QLF) candidate $V(x) = x^T P x$, where $P = P^T \succ 0$ is some matrix. Its time-derivative along the solution of~\eqref{eq:nonlinear-system} is given as
\begin{align*}
\dot V(x_t) = \sum_{i = 1}^r {{\alpha _i}(x_t) x_t^T(A_i^TP + P{A_i})x_t}
\end{align*}
Using Lyapunov direct method~\cite{khalil2002nonlinear}, the corresponding stability condition is given by the following Lyapunov matrix inequality:
\begin{align}
P \succ 0,\quad A{(\alpha (x))^T}P + PA(\alpha (x)) \prec 0,\quad x \in {\cal X}.\label{eq:0}
\end{align}

Using the fact that $A(\alpha (x)) \in \{ A(\alpha ) \in {\mathbb R}^{n \times n}:\alpha  \in {\Lambda _r}\}  = {\bf{co}}\{ {A_1},{A_2}, \ldots ,{A_r}\}$, one can derive a sufficient condition to guarantee the Lyapunov matrix inequality in~\eqref{eq:0} summarized below.
\begin{lemma}[{\cite[Theorem~5]{tanaka2004fuzzy}}]\label{lemma:quadratic}
Suppose that there exists a symmetric matrix $P = P^T \in {\mathbb R}^{n\times n}$ such that the following PLMIs hold:
\begin{align}
P \succ 0,\quad A{(\alpha )^T}P + PA(\alpha ) \prec 0,\quad \forall \alpha  \in {\Lambda _r}.\label{eq:3}
\end{align}
Then, the fuzzy system~\eqref{eq:fuzzy-system} is locally asymptotically stable, and a DA is given by any Lyapunov sublevel sets defined by ${L_V}(c): = \{ x \in {\mathbb R}^n:V(x) \le c\}  \subseteq {\cal X}$ with any $c>0$ such that ${L_V}(c) \subseteq {\cal X}$. The largest DA is ${L_V}(c^*)$ with ${c^*}: = \max \{ c > 0:{L_V}(c) \subseteq {\cal X} \}$.
\end{lemma}
\begin{proof}
It is straightforward to prove that~\eqref{eq:3} implies~\eqref{eq:0}. Then, the Lyapunv method~\cite{khalil2002nonlinear} can be applied to complete the proof.
\end{proof}

Note that~\cref{lemma:quadratic} is a local stability condition in the sense that the DA is a strict subset of the modeling region $\cal X$.
However, traditionally it was called a global condition in the sense that the Lyapunov inequality
\begin{align*}
V(x) > 0,\quad {\nabla _x}V{(x)^T}A(\alpha (x))x < 0,\quad \forall x \in {\cal X}\backslash \{ 0\}
\end{align*}
hold for all $x \in {\cal X}\backslash \{ 0\}$ inside the modeling region $\cal X$.
In~\cref{lemma:quadratic}, the infinite-dimensional PLMIs can be converted to finite-dimensional sufficient LMI condition using~\cref{thm:relax0}.

While the LMI condition offers simplicity, a key drawback of the quadratic stability approach is its inherent conservatism. This conservatism primarily stems from the requirement for a common Lyapunov matrix $P$, which should satisfy the Lyapunov inequality across all subsystems of the fuzzy system. As an alternative, FLFs can be employed.

\subsection{Fuzzy Lyapunov stability}

Let us consider the Lyapunov function candidate
\begin{align}
V(x) = {x^T}\left( {\sum_{i = 1}^r {{\alpha _i}(x)P_i} } \right)x =: {x^T}P(\alpha(x))x,\label{eq:fuzzy-Lyapunov-function}
\end{align}
which is called a FLF, where $P_i = P_i^T \succ 0,i\in {\cal I}_r$.
Taking its time-derivative along the trajectory of~\eqref{eq:fuzzy-system} leads to
\begin{align*}
\dot V(x_t) =x_t^T\left[ {\sum_{i = 1}^r {\sum_{j = 1}^r {{\alpha _i}(x_t){\alpha _j}(x_t)(A_i^T{P_j} + {P_j}{A_i})} } } + {\sum_{k = 1}^r {{{\dot \alpha }_k}(x_t){P_k}} } \right]x_t.
\end{align*}

It is important to note that the use of the FLF presupposes that the MFs are differentiable, an assumption that does not universally hold. Therefore, when employing the FLF, we implicitly assume the differentiability of the MFs. A challenge that emerges in this context is that the stability condition depends on the time-derivatives of the MFs along the solution of~\eqref{eq:fuzzy-system}, which are not readily amenable to convexification. To circumvent this issue, \cite{tanaka2003multiple} proposed incorporating a bound on the time-derivatives of the MFs, $|{\nabla _x}{\alpha _i}(x)A(\alpha (x))x| \le {\varphi _i},i \in {\cal I}_r,x \in {\cal X}$. Building on this bound, a subsequent stability condition is derived as follows.
\begin{lemma}[{\cite[Theorem~1]{tanaka2003multiple}}]\label{lemma:Tanaka2003}
Suppose that the MFs are differentiable.
Moreover, suppose that there exist constants $\varphi_k > 0, k\in {\cal I}_r$, symmetric matrices $P_i = P_i^T \in {\mathbb R}^{n\times n}, i\in {\cal I}_r$ such that the following PLMIs hold:
\begin{align*}
P(\alpha ) \succ 0,\quad \sum_{k = 1}^r {{\varphi _k}{P_k}}  + P(\alpha )A(\alpha ) + A{(\alpha )^T}P(\alpha ) \prec 0
\end{align*}
for all $\alpha  \in {\Lambda _r}$. Then, the fuzzy system~\eqref{eq:fuzzy-system} is locally asymptotically stable, and any Lyapunov sublevelset ${L_V}(c): = \{ x \in {\mathbb R}^n:V(x) \le c\}$ with some $c>0$ such that ${L_V}(c) \subseteq  \Omega (\varphi )$ is a DA, where
\begin{align}
\Omega (\varphi ): = \{ x \in {\cal X}:|{\nabla _x}{\alpha _i}{(x)^T}A(\alpha (x))x| \le {\varphi_i},i \in {\cal I}_r\}\label{eq:set-derivative-bound}
\end{align}
and $\varphi : = {\left[ {\begin{array}{*{20}{c}}
{{\varphi _1}}&{{\varphi _2}}& \cdots &{{\varphi _r}}
\end{array}} \right]^T}$.
\end{lemma}

In~\cref{lemma:Tanaka2003}, $\Omega (\varphi )$ is a subset of $\cal X$ such that the time-derivatives of the MFs are bounded within certain intervals. If the LMIs in~\cref{lemma:Tanaka2003} are feasible, then we have
\begin{align*}
&V(x)> 0,\quad \forall x \in {\cal X} \backslash \{ 0\},\\
&{\nabla _x}V{(x)^T}A(\alpha (x))x < 0,\quad \forall x \in \Omega (\varphi )\backslash \{ 0\}.
\end{align*}
Here, note that the Lyapunov inequality holds only within $\Omega (\varphi )\backslash \{ 0\}$ rather than ${\cal X} \backslash \{ 0\}$ due to the constraints on the time-derivatives of the MFs. This is one of the differences between the QLF approach and FLF approach. In the FLF approach, the imposed bounds $|{\nabla _x}{\alpha _i}{(x)^T}A(\alpha (x))x| \le {\varphi _i},i \in {\cal I}_r$ imply that the matrix $A(\alpha(x_t))$ varies at a sufficiently slow rate along the solution of~\eqref{eq:fuzzy-system}.

In~\cite{mozelli2009reducing}, an improved (less conservative) LMI condition has been proposed using the null effect
\begin{align}
\sum\limits_{k = 1}^r {{{\dot \alpha }_k}(x_t)}  = 0.\label{eq:slack1}
\end{align}
Using this property, one can add the slack term $\sum_{k = 1}^r {{{\dot \alpha }_k}(x_t)}M   = 0$ with any slack matrix variable $M$. For convenience and completeness of the presentation, the condition is formally introduced below.
\begin{lemma}[{\cite[Theorem~6]{mozelli2009reducing}}]\label{lemma:Mozelli2009}
Suppose that the MFs are continuously differentiable.
Moreover, suppose that there exist constants $\varphi_k > 0,k\in {\cal I}_r$, symmetric matrices $P_i = P_i^T \in {\mathbb R}^{n\times n}, i\in {\cal I}_r$ and $M = M^T \in {\mathbb R}^{n\times n}$ such that the following LMIs hold:
\begin{align}
&P(\alpha) \succ 0,\quad P(\alpha) + M \succeq 0,\quad \forall \alpha  \in {\Lambda _r}\nonumber\\
&\sum\limits_{k = 1}^r {{\varphi _k}({P_k} + M)}  + P(\alpha )A(\alpha ) + A{(\alpha )^T}P(\alpha ) \prec 0,\quad \forall \alpha  \in {\Lambda _r}\label{eq:4}
\end{align}
Then, the fuzzy system~\eqref{eq:fuzzy-system} is locally asymptotically stable, and any Lyapunov sublevelset ${L_V}(c)$ with some $c>0$ such that ${L_V}(c) \subseteq \Omega(\varphi)$ is a DA.
\end{lemma}

As $\varphi_i$ approaches zero for $i\in {\cal I}_r$, the associated PLMI condition in~\cref{lemma:Mozelli2009} become less conservative. However, this reduction in conservatism comes at a cost: the set $\Omega(\varphi )$ contracts towards the origin, leading to a corresponding reduction in the DA. This situation shows a trade-off between the level of conservatism in the stability conditions and the volume of the DA. In essence, achieving less conservative conditions results in a more restricted DA.

\section{Augmented system}\label{sec:augmented-system}

\subsection{Fuzzy model of Jacobian}

In this paper, we propose a new stability condition, which is derived based on a new augmented Lyapunov function that is an extension of the FLFs. A main feature of the conditions is that they do not depend on the bounds on the time-derivatives of the MFs. Instead, we consider the Jacobian of the MF vector in~\eqref{eq:alpha-vec} as follows:
\begin{align*}
J(x): = {\nabla _x}\alpha (x) \in {\mathbb R}^{r \times n},\quad x \in {\cal X}.
\end{align*}

In our stability analysis, we will use the fuzzy modeling of the Jacobian
\begin{align*}
J(x) = \sum_{i = 1}^p {{\beta _i}(x){J_i}}  = :J(\beta (x)),\quad x \in {\cal X}
\end{align*}
where $J_i \in {\mathbb R}^{r \times
n}, i\in \{1,2,\ldots, p \} = {\cal I}_p$, and the vector
\begin{align*}
\beta (x): = \left[ {\begin{array}{*{20}{c}}
{{\beta_1}(x)}\\
{{\beta_2}(x)}\\
 \vdots \\
{{\beta_p}(x)}
\end{array}} \right] \in {\Lambda _p}
\end{align*}
is the MF corresponding to the Jacobian. One can say that the Jacobian matrix varies within the convex set
\[J(x) \in {\bf{co}}\{ {J_1},{J_2}, \ldots ,{J_p}\} ,\quad x \in {\cal X}\]
where ${\bf{co}}\{  \cdot \} $ is the convex hull~\cite{boyd1994linear}.
Note that here, we assume that one can obtain the fuzzy model of the Jacobian matrix within the same modeling region $\cal X$ of the underlying fuzzy system, which is the case in most instances.
\begin{example}\label{ex:2}
Let us consider the system in~\eqref{ex:1} again. The Jacobian matrix of $\alpha(x)$ in~\eqref{eq:alpha-vec} is
calculated as follows:
\begin{align*}
&J(x)  = \frac{1}{2}\cos x_1 \left[ {\begin{array}{*{20}c}
   1 & 0  \\
   { - 1} & 0  \\
\end{array}} \right] \in {\rm{co}}\{ J_1 ,\,J_2 \} ,\quad \forall x \in {\cal X},
\end{align*}
where
\begin{align*}
&J_1  = \left[ {\begin{array}{*{20}c}
   0 & 0  \\
   0 & 0  \\
\end{array}} \right],\quad J_2  = \left[ {\begin{array}{*{20}c}
   {0.5} & 0  \\
   { - 0.5} & 0  \\
\end{array}} \right].
\end{align*}
\end{example}

\subsection{Augmented system model}

The time-derivative of the MF in~\eqref{eq:alpha-vec} can be written as
\begin{align}
\dot \alpha ({x_t}) = J(\beta(x_t))A(\alpha ({x_t})) x_t,\label{eq:membership-system}
\end{align}
which can be interpreted as another nonlinear system that depends on the state $x_t$.
Therefore, the underlying system~\eqref{eq:fuzzy-system} and~\eqref{eq:membership-system} can be seen as an interconnected system.
In particular, by incorporating it into the original system, we can obtain the augmented system
\begin{align*}
\frac{d}{{dt}}\left[ {\begin{array}{*{20}{c}}
{{x_t}}\\
{\alpha ({x_t}) - \alpha (0)}
\end{array}} \right]= \left[ {\begin{array}{*{20}{c}}
{A(\alpha ({x_t}))}&0\\
{J(\beta(x_t))A(\alpha ({x_t}))}&0
\end{array}} \right]\left[ {\begin{array}{*{20}{c}}
{{x_t}}\\
{\alpha ({x_t}) - \alpha (0)}
\end{array}} \right]
\end{align*}
Here, $\alpha ({x_t}) - \alpha (0)$ is used instead of $\alpha ({x_t})$ as state variables in order to make the origin an equilibrium point. Next, defining
\begin{align*}
A(\alpha (x)) =& \sum\limits_{i = 1}^r {{\alpha _i}(x){A_i}}  = \sum\limits_{i = 1}^r {({\alpha _i}(x) - {\alpha _i}(0)){A_i}}  + \sum\limits_{i = 1}^r {{\alpha _i}(0){A_i}} \\
=& \sum\limits_{i = 1}^r {({\alpha _i}(x) - {\alpha _i}(0)){A_i}}  + {A_0}
\end{align*}
we write
\begin{align*}
\frac{d}{{dt}}\left[ {\begin{array}{*{20}{c}}
{{x_t}}\\
{\alpha ({x_t}) - \alpha (0)}
\end{array}} \right]= \left[ {\begin{array}{*{20}{c}}
{\sum\limits_{i = 1}^r {({\alpha _i}({x_t}) - {\alpha _i}(0)){A_i}}  + {A_0}}&0\\
{J(\beta(x_t))\sum\limits_{i = 1}^r {({\alpha _i}({x_t}) - {\alpha _i}(0)){A_i}}  + {A_0}}&0
\end{array}} \right] \left[ {\begin{array}{*{20}{c}}
{{x_t}}\\
{\alpha ({x_t}) - \alpha (0)}
\end{array}} \right]
\end{align*}
which is an interconnection between two nonlinear systems as illustrated in~\cref{fig:2}.
\begin{figure}
\centering
\includegraphics[width=13cm]{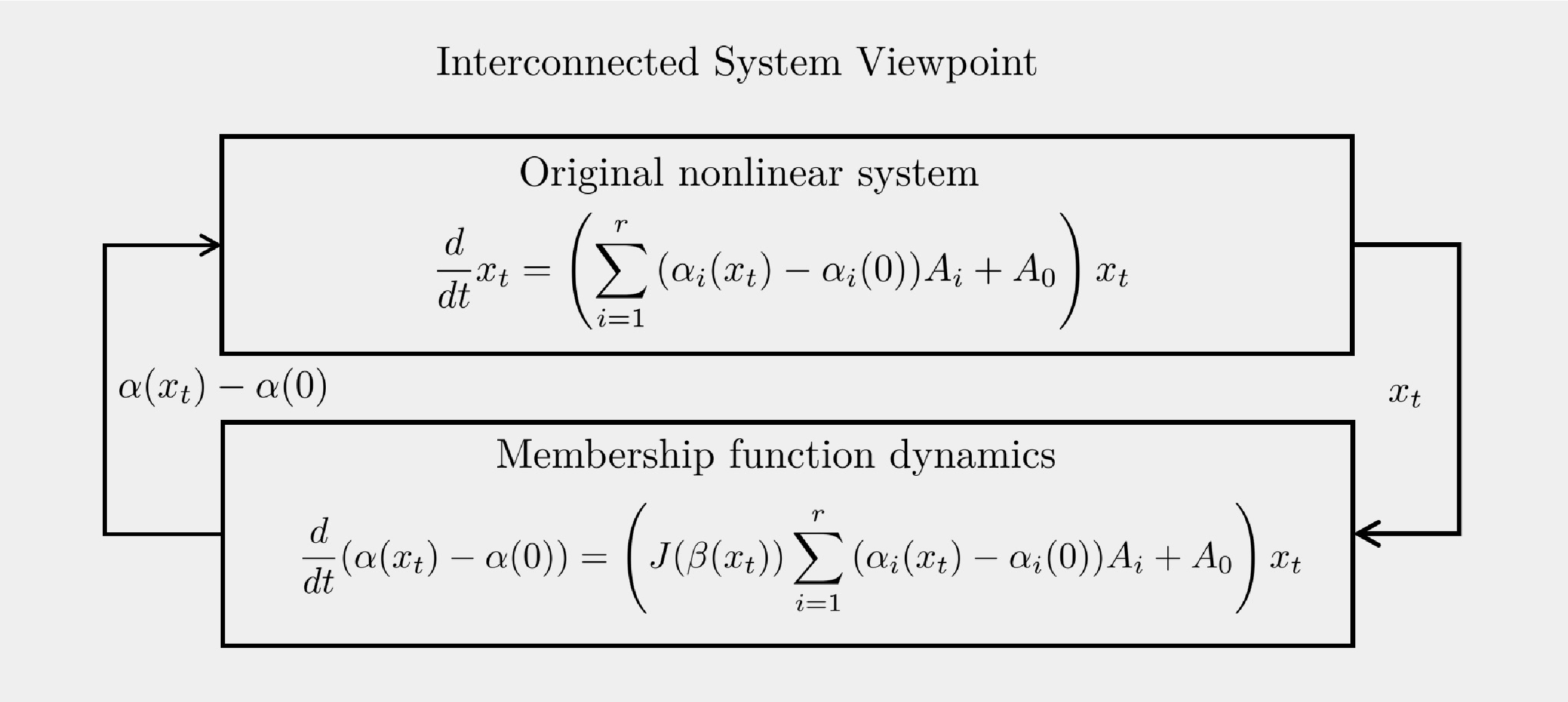}
\caption{Interconnected system viewpoint of T-S fuzzy system}\label{fig:2}
\end{figure}

The system can be viewed as a nonlinear quadratic system in the augmented state
\[\left[ {\begin{array}{*{20}{c}}
{{x_t}}\\
{\alpha ({x_t}) - \alpha (0)}
\end{array}} \right] = :{\xi _t}\]
under the special constraints on the state ${\bf{1}}_r^T (\alpha (x) - \alpha (0)) = 0$.
Therefore, techniques in the stability analysis of the quadratic system can be considered.
While the augmented system mentioned above offers promising potential for the development of new stability analyses, this paper will focus on a different augmented system. However, the perspective of the interconnected system in~\cref{fig:2} provides a basis for the main development in this paper.

In particular, let us define the function $\xi :{\mathbb R}^n \to {{\mathbb R}^{n + rn + {n^2}}}$ as
\[\xi (x): = \left[ {\begin{array}{*{20}{c}}
x\\
{\alpha (x) \otimes x}\\
{x \otimes x}
\end{array}} \right].\]
Then, we will consider the augmented state vector
\begin{align}
{\xi _t}: =\xi (x_t) = \left[ {\begin{array}{*{20}{c}}
{{x_t}}\\
{\alpha ({x_t}) \otimes {x_t}}\\
{{x_t} \otimes {x_t}}
\end{array}} \right] \in {\mathbb R}^{n + rn + n^2}.\label{eq:augmented-state}
\end{align}

Its time-derivative along the solution of~\eqref{eq:fuzzy-system} is given by
\begin{align}
\frac{d}{{dt}}\left[ {\begin{array}{*{20}{c}}
{{x_t}}\\
{\alpha ({x_t}) \otimes {x_t}}\\
{{x_t} \otimes {x_t}}
\end{array}} \right] = \underbrace {\left[ {\begin{array}{*{20}{c}}
{A(\alpha ({x_t}))}&{{0_{n \times rn}}}&{{0_{n \times {n^2}}}}\\
{{0_{rn \times n}}}&{{I_r} \otimes A(\alpha ({x_t}))}&{J(\beta ({x_t}))A(\alpha ({x_t})) \otimes {I_n}}\\
{{0_{{n^2} \times n}}}&{{0_{{n^2} \times rn}}}&{{I_n} \otimes A(\alpha ({x_t})) + A(\alpha ({x_t})) \otimes {I_n}}
\end{array}} \right]}_{ = :\Phi ({x_t}) \in { {\mathbb R}^{(n + rn + {n^2}) \times (n + rn + {n^2})}}}\underbrace {\left[ {\begin{array}{*{20}{c}}
{{x_t}}\\
{\alpha ({x_t}) \otimes {x_t}}\\
{{x_t} \otimes {x_t}}
\end{array}} \right]}_{= :\xi(x_t) = :\xi _t}.\label{eq:augmented-system}
\end{align}

At first glance, the construction of the augmented state vector as shown in~\eqref{eq:augmented-state} may seem unusual. However, we will subsequently elucidate the motivations behind considering this augmented state vector. First of all, the augmented state vector $\xi (x_t)$ is constructed with the following constraints:
\begin{enumerate}
\item It includes both state variables of the original system and the MFs.

\item The corresponding differential equation has the form of the nonlinear system in~\eqref{eq:nonlinear-system}.
\end{enumerate}

To develop a augmented state vector that satisfies the above constraints, we note that
\begin{enumerate}
\item the time-derivative of $x_t$ includes $\alpha ({x_t}) \otimes {x_t}$. Therefore, $\alpha ({x_t}) \otimes {x_t}$ must be involved in the augmented state vector;

\item Moreover, the time-derivative of $\alpha ({x_t}) \otimes {x_t}$ includes both $\alpha ({x_t}) \otimes {x_t}$ and ${x_t} \otimes {x_t}$. Therefore, ${x_t} \otimes {x_t}$ must be involved in the augmented state vector.
\end{enumerate}

For these reasons, we consider the specific augmented state vector in~\eqref{eq:augmented-state} in this paper. Note also that in this formulation, $\alpha(x)$ is no more the MFs, but it can be viewed as a part of the state variables.

The main advantages of considering the new state vector in~\eqref{eq:augmented-state} are summarized as follows:
\begin{enumerate}
\item In this formulation, $\alpha(x)$ is no more MFs, but it can be viewed as a part of the state variables. Therefore, we no more need to consider time-derivatives of MFs which are intricate to deal with because the derivatives are naturally included in the differential equation formulation.

\item Moreover, useful information on MFs is now naturally incorporated in the new nonlinear dynamic system in~\eqref{eq:augmented-system}. Therefore, this can potentially reduce conservativeness of the resulting condition.
\end{enumerate}

It is straightforward to prove that the stability of~\eqref{eq:augmented-system} is equivalent to the stability of~\eqref{eq:fuzzy-system}.
\begin{proposition}\label{stability-identity}
The system~\eqref{eq:fuzzy-system} is locally asymptotically stable if and only if~\eqref{eq:augmented-system} is locally asymptotically stable.
\end{proposition}
\begin{proof}
For the sufficiency, let us assume that~\eqref{eq:augmented-system} is locally asymptotically stable, i.e., with $\xi(x_0) \in {\cal D}$ in some region ${\cal D} \subseteq {\mathbb R}^{n + rn + n^2}$ including the origin $0 \in {\mathbb R}^{n + rn + n^2}$, $\xi(x_t) \to 0$ as $t\to \infty$. Then, this obviously implies that with $x_0 \in {\cal D}'$, $x_t \to 0$ as $t \to \infty$, where the set ${\cal D}' \subseteq {\mathbb R}^n$ is defined as ${\cal D}': = \{ x \in {\cal X}:\xi(x) \in {\cal D}\}$. Therefore,~\eqref{eq:fuzzy-system} is also locally asymptotically stable. Conversely, suppose that~\eqref{eq:fuzzy-system} is locally asymptotically stable, which implies that starting from $x_0 \in {\cal D}$ for some set ${\cal D} \subseteq {\cal X} \subseteq {\mathbb R}^{n}$ including the origin, $x_t \to 0$ as $t\to \infty$. This means that starting from $\xi(x_0) \in {\cal D}'$, where ${\cal D}' = \{ \xi (x) \in {\mathbb R}^{n + rn + {n^2}}:x \in {\cal D}\}$ includes the origin, $\xi _t\to 0$ as $t \to \infty$. Therefore,~\eqref{eq:augmented-system} is locally asymptotically stable. This completes the proof.
\end{proof}

Building on this result, this paper will focus on the stability of \eqref{eq:augmented-system} rather than that of \eqref{eq:fuzzy-system}. This approach offers some benefits, as demonstrated in the subsequent sections.

\section{Main results}\label{sec:stablity}

In this paper, for the augmented system in~\eqref{stability-identity}, we consider the following augmented Lyapunov function candidate:
\begin{align}
V(x) = {\left[ {\begin{array}{*{20}{c}}
x\\
{\alpha (x) \otimes x}\\
{x \otimes x}
\end{array}} \right]^T}P\left[ {\begin{array}{*{20}{c}}
x\\
{\alpha (x) \otimes x}\\
{x \otimes x}
\end{array}} \right] = :{\xi ^T}P\xi  = :V(\xi),\label{eq:Lyapunov-function1}
\end{align}
where $\xi : =\xi(x) : = \left[ {\begin{array}{*{20}{c}}
x\\
{\alpha (x) \otimes x}\\
{x \otimes x}
\end{array}} \right] \in {\mathbb R}^{n + rn + {n^2}}$, $P = P^T \in {\mathbb R}^{(n+rn+n^2) \times (n+rn+n^2)}$ and $P \succ 0$.
We note that it is quadratic in $\xi(x)$, non-quadratic in $x$, and is more general than the FLF in~\eqref{eq:fuzzy-Lyapunov-function}. Its time-derivative along the solution of~\eqref{eq:augmented-system} is given by
\begin{align*}
\dot V(x_t) = \xi(x_t)^T[P\Phi ({x_t}) + \Phi {({x_t})^T}P] \xi(x_t),
\end{align*}
and the corresponding Lyapunov matrix inequality is given by
\begin{align}
P \succ 0,\quad P\Phi (x) + \Phi {(x)^T}P \prec 0,\quad \forall x \in {\cal X}.\label{eq:1}
\end{align}

Since the MFs' dynamics is incorporated in~\eqref{eq:augmented-system}, the corresponding Lyapunov matrix inequality in~\eqref{eq:1} does not involve the time-derivatives of the MFs:~\eqref{eq:1} is a derivative-free condition, and does not require bounds on the time-derivatives of the MFs anymore.

At this point, advantages of considering the augmented system in~\eqref{eq:augmented-system} become clearer, and the related remarks are in order.
\begin{enumerate}
\item First of all, the use of the augmented system in~\eqref{eq:augmented-system} effectively converts the stability problem of~\eqref{eq:nonlinear-system} based on nonlinear FLFs into those based on simpler quadratic Lyapunov function in~\eqref{eq:Lyapunov-function1}.

\item Moreover, although~\eqref{eq:Lyapunov-function1} is a simpler quadratic Lyapunov function, it can be proved that it is more general than FLFs given in~\eqref{eq:fuzzy-Lyapunov-function}.
\end{enumerate}

Next, to find a solution of~\eqref{eq:1} for $P$, one needs to solve the infinite-dimensional matrix inequality in~\eqref{eq:1}, which is in general very hard and numerically intractable.
To resolve this issue, we can first use the special structures
\[A(\alpha ({x_t})) \in \{ A(\alpha ) \in {\mathbb R}^{n \times n}:\alpha  \in {\Lambda _r}\}  = {\bf{co}}\{ {A_1},{A_2}, \ldots ,{A_r}\} \]
and
\[J(\beta ({x_t})) \in \{ J(\beta ) \in {\mathbb R}^{r \times n}:\beta  \in {\Lambda _p}\}  = {\bf{co}}\{ {J_1},{J_2}, \ldots ,{J_p}\} \]
In particular, defining
\begin{align*}
\Phi (\alpha ,\beta ): = \left[ {\begin{array}{*{20}{c}}
{A(\alpha )}&0&0\\
0&{I_r \otimes A(\alpha )}&{J(\beta )A(\alpha ) \otimes I_n}\\
0&0&{I_n \otimes A(\alpha ) + A(\alpha ) \otimes I_n}
\end{array}} \right]
\end{align*}
depending on $\alpha \in \Lambda_r$ and $\beta \in \Lambda_p$, a sufficient condition can be obtained as follows.
\begin{proposition}\label{prop:1}
Suppose that the MFs are continuously differentiable.
Moreover, suppose that there exists a symmetric matrices $P = P^T \in {\mathbb R}^{(n+rn+n^2)\times (n+rn+n^2)}$ such that the following PLMIs hold:
\begin{align}
P \succ 0,\quad P\Phi (\alpha ,\beta ) + \Phi {(\alpha ,\beta )^T}P \prec 0\label{eq:2}
\end{align}
for all $(\alpha ,\beta ) \in {\Lambda _r} \times {\Lambda _p}$.
Then, the fuzzy system~\eqref{eq:fuzzy-system} is locally asymptotically stable, and Lyapunov sublevelset ${L_V}(c): = \{ x \in {\mathbb R}^n:V(x) \le c\}$ with some $c>0$ such that ${L_V}(c) \subseteq  {\cal X}$ is a DA.
\end{proposition}
\begin{proof}
Suppose that the PLMIs in~\eqref{eq:2} hold. Then, $P \succ 0$ implies $V(x) = \xi(x)^T P \xi(x) > 0,\forall x \in {\cal X}\backslash \{ 0\}$.
Moreover, the second condition, $P\Phi (\alpha ,\beta ) + \Phi {(\alpha ,\beta )^T}P \prec 0,\forall (\alpha ,\beta ) \in {\Lambda _r} \times {\Lambda _p}$ implies the Lyapunov matrix inequality $P\Phi (x) + \Phi {(x)^T}P \prec 0,\forall x \in {\cal X}$. Therefore, pre- and post-multiplying the last inequality with $\xi (x)^T$ and its transpose lead to $\xi {(x)^T}(P\Phi (x) + \Phi {(x)^T}P)\xi (x) = {\nabla _x}V{(x)^T}\Phi (x) \xi (x) < 0$ for all $x \in {\cal X}\backslash \{ 0\}$.
Therefore, it ensures that $V(x) > 0,{\nabla _x}V{(x)^T}A(\alpha (x))x < 0,\forall x \in {\cal X}\backslash \{ 0\}$. Then, by the Lyapunov method~\cite{khalil2002nonlinear}, the augmented system~\eqref{eq:augmented-system} is locally asymptotically stable. Finally,~\cref{stability-identity} implies that the fuzzy system~\eqref{eq:fuzzy-system} is locally asymptotically stable as well. This completes the proof.
\end{proof}

Note that the condition~\eqref{eq:2} is still an infinite-dimensional PLMIs over matrix polytopes.
However, for these classes of problems, there are several finite-dimensional LMI relaxations~\cite{tanaka1998fuzzy,tanaka2004fuzzy,kim2000new,Teixeira2003,fang2006new,sala2007asymptotically,sala2007relaxed,oliveira2007parameter,lee2010improvement,kim2023relaxed} with various degrees of conservatism. Therefore,~\eqref{eq:2} can be efficiently solved using these techniques. In particular, since~\eqref{eq:2} are expressed as double fuzzy summations for both $\alpha \in \Lambda_r$ and $\beta \in \Lambda_p$, one can apply~\cref{thm:relax1}.

Although~\cref{prop:1} is promising, an issue should be resolved. Specifically, it can be easily proved that~\cref{prop:1} is not less conservative than the QLF approach in~\cref{lemma:quadratic}. This is attributed to the requirement that for~\eqref{eq:2} to be satisfied, the following condition must be met:
\begin{align}
{P_{11}} \succ 0,\quad A{(\alpha )^T}{P_{11}} + {P_{11}}A(\alpha ) \prec 0,\quad \forall \alpha  \in {\Lambda _r}\label{eq:6}
\end{align}
where $P_{11} = P_{11}^T \in {\mathbb R}^{n\times n}$ is the first $n$-by-$n$ diagonal matrix of $P \in {\mathbb R}^{(n+rn+n^2)\times (n+rn+n^2)}$ in~\eqref{eq:2}. Then, one can observe that~\eqref{eq:6} corresponds to the Lyapunov matrix inequality based on the QLF in~\cref{lemma:quadratic}. This means that for~\eqref{eq:2} to be satisfied, the QLF-based condition in~\cref{lemma:quadratic} should be satisfied. Therefore, the condition in~\cref{prop:1} is no less conservative than the quadratic stability condition in~\cref{lemma:quadratic}. This conclusion is formally stated in the following result.
\begin{proposition}
The condition in~\cref{prop:1} is no less conservative than~the quadratic stability condition in~\cref{lemma:quadratic}.
\end{proposition}

The main source of the conservatism of~\eqref{eq:2} comes from the conversion from~\eqref{eq:1} to~\eqref{eq:2}, which introduces significant conservatism because structural information of the system~\eqref{eq:augmented-system} is lost in~\eqref{eq:2}.
In particular, the major structure information lost is summarized as follows: in~\eqref{eq:2}, $A(\alpha(x))$ is simply replaced with $A(\alpha ),\alpha  \in {\Lambda_r}$, and it loses the fact that $A(\alpha(x))$ depends on the state components $\alpha(x)$:
\begin{align}
A(\alpha ) = \sum_{i = 1}^r {{\alpha _i}{A_i}}\label{eq:structure1}
\end{align}

In what follows, we introduce how to incorporate such structural and geometric properties into the stability conditions with additional slack variables, which eventually lead to less conservative conditions.
In particular,~\eqref{eq:structure1} holds if and only if
\begin{align*}
A(\alpha )x = \sum\limits_{i = 1}^r {{\alpha _i}{A_i}x}  = \left[ {\begin{array}{*{20}{c}}
{{A_1}}&{{A_2}}& \cdots &{{A_r}}
\end{array}} \right](\alpha  \otimes x),\quad
\end{align*}
for all $x \in {\mathbb R}^n \backslash \{ 0 \}$. The above equation can be rewritten as
\begin{align*}
{0_n} = \underbrace {\left[ {\begin{array}{*{20}{c}}
{A(\alpha )}&{ - \left[ {\begin{array}{*{20}{c}}
{{A_1}}&{{A_2}}& \cdots &{{A_r}}
\end{array}} \right]}& 0_{n \times n^2}
\end{array}} \right]}_{: = {\Gamma}(\alpha ) \in {\mathbb R}^{n \times (n + rn + {n^2})}}\underbrace {\left[ {\begin{array}{*{20}{c}}
x\\
{\alpha  \otimes x}\\
{x \otimes x}
\end{array}} \right]}_{ = :\xi(x) = :\xi }
\end{align*}
for all $\xi \in {\mathbb R}^{n+rn+n^2} \backslash \{ 0 \}$.
Therefore, for any slack variable matrix $N(\alpha)\in {\mathbb R}^{(n + rn + {n^2}) \times n}$, we have
\begin{align}
{\xi(x) ^T}[N(\alpha ){\Gamma}(\alpha ) + {\Gamma}{(\alpha )^T}N{(\alpha )^T}]\xi(x)  = 0,\quad \forall \alpha \in \Lambda_r, \beta \in \Lambda_p\label{eq:5}
\end{align}

Therefore, combining~\eqref{eq:5} with the Lyapunov inequality~\eqref{eq:2}, one can derive the following condition.
\begin{proposition}\label{prop:2}
Suppose that the MFs are continuously differentiable.
Moreover, suppose that there exists a symmetric matric $P = P^T \in {\mathbb R}^{(n+rn+n^2)\times (n+rn+n^2)}$, and parameterized matrices $N(\alpha)\in {\mathbb R}^{(n + rn + {n^2}) \times n} $ such that the following PLMIs hold:
\begin{align}
&P \succ 0,\label{eq:7}\\
&P\Phi (\alpha ,\beta ) + \Phi {(\alpha ,\beta )^T}P + N(\alpha ){\Gamma}(\alpha ) + {\Gamma}{(\alpha )^T} N{(\alpha )^T} \prec 0\label{eq:8}
\end{align}
for all $(\alpha ,\beta ) \in {\Lambda _r} \times {\Lambda _p}$. Then, the fuzzy system~\eqref{eq:fuzzy-system} is locally asymptotically stable, and Lyapunov sublevelset ${L_V}(c): = \{ x \in {\mathbb R}^n:V(x) \le c\}$ with some $c>0$ such that ${L_V}(c) \subseteq  {\cal X}$ is a DA.
\end{proposition}
\begin{proof}
The proof is similar to that of~\cref{prop:1} except for the use of~\eqref{eq:5}. Therefore, it is omitted for brevity.
\end{proof}

The stability conditions proposed in this paper thus far are time-derivative-free conditions. The main advantage of these conditions is their simplicity: they eliminate the need to establish bounds on the time-derivatives of MFs, which are hyperparameters requiring adjustment by the analyst.

\section{Comparative analysis}
All numerical experiments in the sequel were treated with the help
of MATLAB 2008a running on a Windows 7 PC with Intel Core i7-3770
3.4 GHz CPU, 24 GB RAM. The LMI problems were solved with SeDuMi
\cite{Strum1999} and Yalmip \cite{Lofberg2004}.

Let us consider the system in~\eqref{ex:1} in~\cref{ex:1} and~\cref{ex:2} again. In this example, we will compare the following methods: 1) the QLF-based approach in~\cref{lemma:quadratic}, 2) the FLF-based approach in~\cref{lemma:Tanaka2003} and~\cref{lemma:Mozelli2009}, the proposed approach~\cref{prop:2}, and Theorem~3 in~\cite{rhee2006new} (line integral approach).

To deal with the PLMIs with single fuzzy summations, the relaxation in~\cref{thm:relax0} is applied.
Moreover, for the PLMIs in~\cref{prop:2} with double fuzzy summations, the relaxation in~\cref{thm:relax1} is applied.
The single fuzzy summation in $\beta$ also arises, and the related PLMI is relaxed with~\cref{thm:relax0}.
For~\cref{prop:2}, we consider the parameterized matrix with single fuzzy summations $N(\alpha ) = \sum_{i = 1}^r {{\alpha _i}{N_{i}}}$.

\begin{table}[h!]
\caption{$\lambda^*$ for different approaches.}
\begin{center}
\begin{tabular}{c c}
\hline

\centering  Methods & $\lambda^*$\\

\hline

\centering \cref{lemma:quadratic} & $3.8269$\\
\centering \cref{prop:2} & $5.4749$\\
\centering \cref{lemma:Tanaka2003} with $\phi_i  = 2,i\in {\cal I}_r$ & $0.0061$\\
\centering \cref{lemma:Tanaka2003} with $\phi_i  = 1,i\in {\cal I}_r$ & $0.0061$\\
\centering \cref{lemma:Tanaka2003} with $\phi_i  = 0.5,i\in {\cal I}_r$ & $0.0061$\\
\centering \cref{lemma:Tanaka2003} with $\phi_i  = 0.1,i\in {\cal I}_r$ & $41.8152$\\
\centering \cref{lemma:Mozelli2009} with $\phi_i  = 2,i\in {\cal I}_r$ & $3.8269$\\
\centering \cref{lemma:Mozelli2009} with $\phi_i  = 1,i\in {\cal I}_r$ & $6.7810$\\
\centering \cref{lemma:Mozelli2009} with $\phi_i  = 0.5,i\in {\cal I}_r$ & $12.9333$\\
\centering \cref{lemma:Mozelli2009} with $\phi_i  = 0.1,i\in {\cal I}_r$ & $53.0457$\\
\centering Theorem~3 in~\cite{rhee2006new} & $7.7454$\\
\hline
\end{tabular}
\label{table:1}
\end{center}
\end{table}

To compare the conservatism, the maximum $\lambda >0$, denoted by $\lambda^*$, such that each method is feasible is calculated using the line search method. In particular, we first set the initial range of $\lambda \in [L,U]$ and check the stability condition for $\lambda = \frac{a+b}{2}$. If feasible, then $L$ is replaced with $\lambda$. Otherwise, $U$ is replaced with $\lambda$. Then, the previous feasibility test is repeated while $U-L > \varepsilon$, where $\varepsilon >0$ is a sufficiently small number that serves as a tolerable error margin. The conservatism of each method can be evaluated by comparing $\lambda^*$, where a larger $\lambda^*$ implies that the corresponding method is less conservative.
Overall comparative results are summarized in~\cref{table:1}.
\begin{remark}
If an LMI solver that is based on interior point methods is used, as, for instance, the LMI control toolbox~\cite{gahinet1996lmi}, the complexity
of the LMI optimization problem can be estimated as being proportional to $N_d^3 N_l$ , where $N_d$ is the total number of scalar decision variables, and $N_l$ is the total row size of the LMIs. \cref{table:2} lists the complexities associated with several approaches given in this paper. Moreover, \cref{fig:3} depicts the estimated computational costs, $\ln (N_d^2{N_l})$, of several approaches for different cases $n=2,r=3$, $n=3,r=3$, and $n=4,r=3$ in log scale. As can be seen from~\cref{fig:3}, the computational cost of~\cref{prop:2} is higher than the other approaches.

\end{remark}

\begin{table}[h!]
\caption{Computational complexities for different approaches.}
\begin{center}
\begin{tabular}{c c c}
\hline

\centering  Methods & $N_d$ & $N_l$\\

\hline

\centering \cref{lemma:quadratic} & $\frac{n(n+1)}{2}$ & $n+nr$\\
\centering \cref{lemma:Tanaka2003} & $\frac{n(n+1)r}{2}$ & $nr + \frac{nr(r+1)}{2}$ \\
\centering \cref{lemma:Mozelli2009} & $\frac{(r+1)n(n+1)}{2}$ & $2nr + \frac{nr(r+1)}{2}$ \\
\centering Theorem~3 in~\cite{rhee2006new} & $n^2-n+nr+\frac{n(n+1)}{2}$ & $2nr + \frac{nr(r-1)}{2}$\\
\centering \cref{prop:2} & $\frac{{(n + rn + {n^2})(n + 3rn + {n^2} + 1)}}{2}$ & $\frac{{({r^3} + {r^2} + 2)(n + rn + {n^2})}}{2}$ \\

\hline
\end{tabular}
\label{table:2}
\end{center}
\end{table}
\begin{figure}[h!]
\centering
\includegraphics[width=13cm]{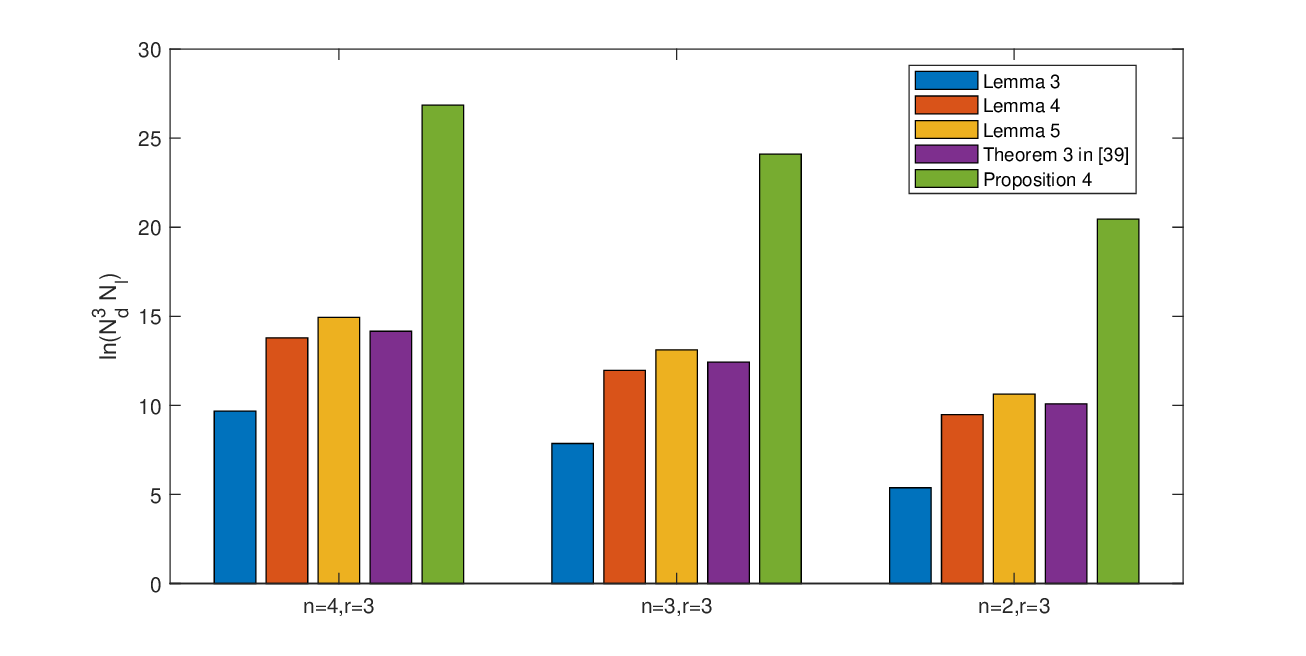}
\caption{Estimated computational cost $\ln (N_d^2{N_l})$ of several approaches for different cases $n=2,r=3$, $n=3,r=3$, and $n=4,r=3$}\label{fig:3}
\end{figure}

Some observations and conclusions from the results are summarized below.
\begin{enumerate}
\item For~\cref{lemma:quadratic}, we have $\lambda^* = 3.8269$, while with~\cref{prop:2}, we have $\lambda^* = 5.4749$.
This result reveals that~\cref{prop:2} is less conservative than the quadratic stability condition in~\cref{lemma:quadratic}.

\item It can be observed that the FLF approaches in~\cref{lemma:Tanaka2003} and~\cref{lemma:Mozelli2009} have the potential to yield less conservative results under the assumption that the time-derivatives of the MFs are bounded. Additionally,~\cref{lemma:Tanaka2003} and~\cref{lemma:Mozelli2009} tend to provide less conservative outcomes when the bounds $\phi_i,i\in {\cal I}_r$ on the time-derivatives of MFs are small. However, this does not imply that~\cref{lemma:Tanaka2003} and~\cref{lemma:Mozelli2009} are inherently less conservative than~\cref{prop:2}. This is because the DA estimates obtained through these methods are restricted to a subset that adheres to the bounds on the time-derivatives of MFs.

\item Theorem~3 in~\cite{rhee2006new} is less conservative than~\cref{prop:2}. However, a notable advantage of~\cref{prop:2} over the line integral approach in~\cite{rhee2006new} is that~\cref{prop:2} is based on a Lyapunov function that does not incorporate an integral. Consequently, this generally can potentially simplify the analysis of local stability.
\end{enumerate}

In summary, it can be concluded that the proposed~\cref{prop:2} is less conservative than the quadratic approach outlined in~\cref{lemma:quadratic}, yet it is more conservative compared to the line integral approach described in~\cite{rhee2006new}. However, the proposed method offers certain advantages over the line integral approach, primarily because it does not involve integration, thus simplifying the analysis of local stability. When comparing to the FLF approaches in~\cref{lemma:Tanaka2003} and~\cref{lemma:Mozelli2009}, a direct comparison of conservativeness is not feasible. Nonetheless, the proposed condition in~\cref{prop:2} offers more benefits as it eliminates the need to adjust the bounds on the time-derivatives of MFs, making it more user-friendly.

\section*{Conclusion}
In this paper, we have developed a new LMI condition for the asymptotic stability of T-S fuzzy systems. A key advantage of this new condition is its independence from the bounds on the time-derivatives of the MFs. This is achieved by introducing a novel FLF that incorporates an augmented state vector, which encompasses the MFs, thus allowing the dynamics of these MFs to be integrated into the proposed condition. The inclusion of additional information about the MFs serves to reduce the conservativeness of the suggested stability condition and also allows the condition to be derivative-free. The outlined approach suggests future research directions including: 1) stability conditions based on the sum-of-square techniques; 2) stabilization conditions using similar augmented state vector and Lyapunov function; 3) relaxations based on the convex sum relaxation techniques or multidimensional convex summation techniques; 4) estimation of DAs. In particular, the augmented state vector in this paper can act as new variables, and they can be naturally integrated into the sum-of-square framework to further reduce the conservatism. The main limitation of the proposed approach in terms of theoretical aspects is that the condition can only address stability problems. In the future, it can be also applied to stabilization problems, which require additional efforts for convexification due to inherent nonconvexity. The limitation of the approach in terms of applications is that it needs for the additional fuzzy modeling of the Jacobian, which requires additional steps and may not be available when the membership functions are complex or non-differentiable. Therefore, future research efforts can be made to overcome these drawbacks.
\section*{Acknowledgement}
The authors would like to thank the Associate Editor and the anonymous Reviewers for their careful reading and
constructive suggestions. This material was supported by the Institute of Information communications Technology Planning Evaluation (IITP) grant funded by the Korea government (MSIT)(No. 2022-0-00469).

\bibliographystyle{IEEEtran}
\bibliography{reference}

\end{document}